\documentclass[11pt]{article}
\normalsize
\usepackage{relsize}
\usepackage{cite}
\usepackage{color}
\usepackage{hyperref}
\hypersetup{
  colorlinks   = true, 
  urlcolor     = green, 
  linkcolor    = blue, 
  citecolor   = red 
}
\usepackage{amsmath,amsthm,amssymb}

\makeatletter
\usepackage{comment}
\let\wfs@comment@comment\comment
\let\comment\@undefined

\usepackage{changes}
\let\wfs@changes@comment\comment
\let\comment\@undefined

\newcommand\comment{%
    \ifthenelse{\equal{\@currenvir}{comment}}
    {\wfs@comment@comment}
    {\wfs@changes@comment}%
}

\definechangesauthor[name=Daniele, color=red]{DAN}
\definechangesauthor[name=Martino, color=blue]{MAR}

\usepackage{stmaryrd}

\newcommand{\PAut}{\mathrm{PAut}}

\newcommand{\F}{\mathbb{F}}

\newcommand{\ZZ}{{\rm Z}}
\newcommand{\B}{{\rm B}}
\newcommand{\HH}{{\rm H}}

\newcommand{\kernel}{\mbox{\rm Ker \,}}

\newcommand{\ho}{\mbox{\rm Hom}}

\newcommand{\Alt}{\mbox{\rm A}}
\newcommand{\Aut}{\mbox{\rm Aut}}
\newcommand{\cha}{\mbox{\rm char\,}}

\DeclareMathOperator{\rank}{rank}

\newcommand{\wt}{\mbox{\rm wt}}

\newcommand{\gene}{\mbox{\rm gen}}
\newcommand{\supp}{\mbox{\rm supp}}

\newcommand{\dist}{\mbox{\rm d}}

\newcommand{\ord}{\mbox{\rm ord}}

\newcommand{\ann}{\mbox{\rm Ann}}

{\theoremstyle{plain}
\newtheorem{theo}{Theorem}[section]
\newtheorem{lem}[theo]{Lemma}
\newtheorem{coro}[theo]{Corollary}
\newtheorem{prop}[theo]{Proposition}}

{\theoremstyle{definition}
\newtheorem{rk}[theo]{Remark}
\newtheorem{defi}[theo]{Definition}
\newtheorem{ex}[theo]{Example}
\newtheorem{question}[theo]{Question}}

\title{Twisted skew G-codes}

\author{Angelot Behajaina, Martino Borello, Javier de la Cruz, Wolfgang Willems}
\date{}

\usepackage{authblk}
\author[1]{Angelot Behajaina}
\affil[1]{Department of Mathematics, Technion - Israel Institute of Technology, Haifa, Israel}

\author[2]{Martino Borello}
\affil[2]{Universit\'e Paris 8, Laboratoire de G\'eom\'etrie, Analyse et Applications, LAGA,
Universit\'e Sorbonne Paris Nord, CNRS, UMR 7539, France}

\author[3]{Javier de la Cruz}
\affil[3]{Department of Mathematics, Universidad del Norte, Barranquilla, Colombia}

\author[3,4]{Wolfgang Willems}
\affil[4]{Otto-von-Guericke-Universit\"at, Magdeburg, Germany}

\begin{document}

\maketitle

\noindent

\begin{abstract} In this paper we investigate left ideals as codes in twisted skew group rings. The considered rings, which are often  algebras over a finite field, allows us to detect many of the well-known codes. The presentation, given here, unifies the concept of group codes, twisted group codes and skew group codes.
\end{abstract}


\section*{Introduction}

Linear codes are fundamental objects in classical coding theory and they are simply subspaces of a a vector space endowed with the Hamming metric. However, the linear structure is usually not enough to get good and useful codes. This is the main reason why more algebraic structure is added since the early days of the theory \cite{Pr,Ber,MacW}. Many interesting codes have found an algebraic realization as group codes \cite{W21},
twisted group codes \cite{CW21} or skew group codes \cite{BGU07}, where the ambient spaces are group algebras,
twisted group algebras, or skew group rings, respectively. For example, the binary extended self-dual $[24,12,8]$ Golay code is a group code,
the ternary extended self-dual $[12,6,6]$ Golay code is a twisted group code and skew cyclic codes are skew group codes.
It turns out that the famous quaternary Hermitian self-dual $[6,3,4]$ hexacode ${\cal H}_6$ is also a skew
group code. Moreover, group codes have been shown to be asymptotically good \cite{BM,BW}.

In this note we define and analyze a unified ambient space
$K[G,\Theta,\alpha]$ 
in which we obtain all classes of the above codes as left ideals via specializations of the parameters
$\Theta$ and $\alpha$. Here $K$ is a finite field, $G$ a finite group, $\Theta \in \ho(G,\Aut(K))$ 
and $\alpha$ a $2$-cocyle of $G$ stabilized by $\Theta(G)$. The object $K[G,\Theta, \alpha]$ is a ring, in general not a $K$-algebra, and called a twisted skew group ring. Specializing the parameters $\Theta$ and $\alpha$ we may obtain a group algebra via $KG=K[G,1,1]$, a twisted group algebra via $K^\alpha G = K[G,1,\alpha]$ or a skew group ring via $K{\rtimes_\Theta}G=K[G,\Theta,1]$. Thus the advantage of the concept are unified proofs for coding theoretical properties of group codes, twisted group codes and skew group codes. \\

The paper is organized as follows. In the first section we define twisted skew group rings  $R=K[G,\Theta,\alpha]$
and investigate the structure of $R$.
In particular, we prove that $R$ is a Frobenius ring. Moreover, $R$ is semisimple if the characteristic of the underlying field does not divide the order of $G$. Section 2 is devoted to left ideals $C$ of $R$ which we call twisted skew group codes. We show that for cyclic groups we get nothing else then the skew $\lambda$-constacyclic codes. Furthermore, we extend a known result on group codes in which the minimum distance of $C$ is bounded from below in terms of $|G|$ and $\dim C$.
In section 3, we answer, in terms of automorphisms, the question when a given linear code is a twisted skew group code.
Section 4 deals with duality questions and section 5 with  codes which are generated by idempotents. In the last section we collect some twisted skew group codes which are optimal according to Grassl's list. In particular we realize the quaternary hexacode as a skew group code.

\section{Twisted skew group rings}

Let $G$ be a finite group, let $K = \F_q$ be a finite field of characteristic $p$ and let $\Theta \in \ho(G,\Aut(K))$.
Furthermore $\alpha$ always denotes a {\it normalized $2$-cocyle} (to be brief a cocyle) of $G$,  i.e.,\\
$\alpha: G \times G \rightarrow K^*=K\setminus \{0\}$ with $\alpha(g,1)=\alpha(1,g) = \alpha(1,1)=1$ for all $g \in G$ satisfying
$$ \alpha(g_1,g_2g_3)\alpha(g_2,g_3)=\alpha (g_1g_2,g_3)\alpha(g_1,g_2)$$
for all $g_i \in G$. We denote the set of all $2$-cocycles for $G$  by $\ZZ^2(G,K^*)$. A $2$-cocycle $\alpha$ is called a {\it coboundary} if there exists $\kappa:G \longrightarrow K^*$ with $\kappa(1)=1$ such that
$$ \alpha(g,h) =  \kappa(g)^{-1} \kappa(h)^{-1} \kappa(gh)$$
for all $g,h\in G$. Let $\textrm{B}^2(G,K^*)$ denote the set of all coboundaries for $G$.

\begin{rk} \label{cohomology} The set $\textrm{B}^2(G,K^*)$ of coboundaries forms a subgroup of the group $\textrm{Z}^2(G,K^*)$
and the factor group $\textrm{H}^2(G,K^*) = \textrm{Z}^2(G,K^*)/\B^2(G,K^*)$ is usually called the {\it second cohomology group of $G$} with values in $K^*$.
 For $\alpha \in \ZZ^2(G,K^*)$ we denote its image in $\HH^2(G,K^*)$ by
$[\alpha]$. Note that $[\alpha]=[\beta]$ if and only if $\alpha(g,h) =\beta(g,h)\kappa(g)\kappa(h)\kappa(gh)^{-1}$ where $\kappa(1)=1$.

\end{rk}
In the sequel, we will always consider cocycles
$\alpha$, which are stabilized by the group 
$\Theta(G) \leq \Aut(K)$, i.e.,
\begin{equation} \label{eq1} \Theta(g)\alpha(x,y) = \alpha(x,y) \end{equation}
for all $g,x,y \in  G$.

\begin{defi} Let
$R:=K[G,\Theta,\alpha]$ be the (left) $K$-vector space  $KG$. On $R$, where the group elements are written with an overline,
 we define a multiplication by

\begin{equation}\label{eq:proddef}
\sum_{g \in G} a_g \overline{g} \cdot \sum_{h \in H} b_h \overline{h}=\sum_{g,h \in G}  a_g(\Theta(g)b_h)\alpha(g,h)\overline{gh}
\end{equation}
for $a_g, b_h \in K$.

\end{defi}

\begin{theo} $R$ is an associative ring with identity $\overline{1}$, called a {\rm twisted skew group ring}.
\end{theo}
\begin{proof}
To prove that $R$ is an associative ring, the only point that needs to be checked is the associativity, 
since the others are straightforward. To this end,
 let $a=\sum_{g}a_g\overline{g} \in R$, $b=\sum_{h}b_h\overline{h} \in R$ and $c=\sum_{u}c_u\overline{u} \in R$. On one hand side we have
\begin{align*}
(a\cdot b)\cdot c&=\left( \sum_{g,h}a_g(\Theta(g)b_h)\alpha(g,h) \overline{gh} \right) \left( \sum_u c_u \overline{u}\right)\\
&=\sum_{g,h,u} a_g(\Theta(g)b_h)\alpha(g,h)(\Theta(gh)c_u)\alpha(gh,u) \overline{(gh)u},
\end{align*}
on the other hand (using that $\Theta(g)$ acts trivially on $\alpha$ \eqref{eq1})
\begin{align*}
a\cdot (b \cdot c)&= \left( \sum_g a_g \overline{g}\right)
\left( \sum_{h,u}b_h(\Theta(h)c_u)\alpha(h,u) \overline{hu} \right) \\
&= \sum_{g,h,u} a_g\Theta(g)[b_h(\Theta(h)c_u)\alpha(h,u)]\alpha(g,hu) \overline{g(hu)}\\
&= \sum_{g,h,u} a_g(\Theta(g)b_h)(\Theta(gh)c_u)\alpha(h,u)\alpha(g,hu) \overline{g(hu)},
\end{align*}
from which  the associativity follows since $\alpha$ satisfies the cocycle condition.
Finally, $\overline{1}$ is the identity since $\alpha$ is assumed to be normalized.
\end{proof}

\begin{rk} {\rm
Note that $K[G,\Theta,\alpha]$ carries the structure of a natural $K$-left vector space since
$$k\cdot \sum_{g \in G} a_g\overline{g} =
k\overline{1}\cdot \sum_{g \in G} a_g\overline{g} =
\sum_{g \in G} k(\Theta(1)a_g)\overline{g} =
\sum_{g \in G} ka_g\overline{g}. 
$$
Furthermore, $K[G,\Theta,\alpha]$ is in general not a $K$-algebra, since 
$\overline{g}(k\overline{1}) = (\Theta(g)k)\overline{g} \neq k \overline{g}$
for $g \in G$ and $k \in K$, if $\Theta(g)k \not= k$.}
\end{rk}

\begin{lem}\label{lem:isogroupcod}
Let $\Theta, \Theta' \in \ho(G,\Aut(K))$ and let $\alpha,\alpha ' \in \ZZ^2(G,K^*)$. If $\delta:G \rightarrow G$ and $\kappa: G \rightarrow K^*$ with $\kappa(1)=1$, then 
\begin{equation} \label{eq}
\Psi: K[G,\Theta,\alpha] \rightarrow K[G,\Theta',\alpha'], \ \  a_g \overline{g}\mapsto a_g\kappa(g) \overline{\delta(g)}
\end{equation} defines a $K$-left distance preserving isomorphism if and only  if 
\begin{itemize}
\item[\rm a)] $\delta$ is a group automorphism satisfying $\Theta=\Theta'  \delta$, and
\item[\rm b)]  $\alpha(g,h)=\kappa(g)\Theta'(\delta(g))(\kappa(h))\kappa(gh)^{-1}\alpha'_\delta(g,h)$ for all $g,h \in G$, \\ where $\alpha'_\delta(g,h)=\alpha'(\delta(g),\delta(h))$  for all $g,h\in G$.
\end{itemize}
\end{lem}
\begin{proof}
Suppose that the map in (\ref{eq}) is a $K$-left isomorphism. For $g,h \in G$ and $a,b \in K$, we have
$$  
\Psi(a\overline{g} \cdot b \overline{h})=\Psi\left(a(\Theta(g)b)\alpha(g,h) \overline{gh}\right)=a(\Theta(g)b)\alpha(g,h) \kappa(gh) \overline{\delta(gh)}
$$
and
$$ \begin{array}{rcl} 
\Psi(a\overline{g} \cdot b \overline{h}) & = & \Psi(a\overline{g}) \Psi(b \overline{h})=\left(a\kappa(g) \overline{\delta(g)}\right)\cdot\left(b \kappa(h) \overline{\delta(h)}\right) \\[1ex] & = &  a\kappa(g) \Theta'(\delta(g))(b\kappa(h))\alpha'(\delta(g),\delta(h))\overline{\delta(g)\delta(h)}.
\end{array}$$
Thus, $\delta(gh)=\delta(g)\delta(h)$ for $g,h \in G$. Since $\delta$ is clearly bijective, it is actually a group automorphism. Next, taking $a=b=1$ in  the equations above, we get
$$
\alpha(g,h)=\kappa(g)\Theta'(\delta(g))\kappa(h)\kappa(gh)^{-1} \alpha'_\delta(g,h),
$$ for $g,h \in G$. Similarly taking $a=1$ and $h=1$  we may deduce that $\Theta(g)(b)=\Theta'(\delta(g))(b)$ for all $g \in G$ and $b \in \mathbb{F}_q$, which implies $\Theta=\Theta'\delta$. \\
The converse follows from the computations above.
\end{proof}

 For a finite ring $A$, we write 
 $\widehat{A}={\rm Hom}_{\mathbb{Z}}(A,\mathbb{C}^\times)$. Recall that $\widehat{A}$ becomes a right $A$-module by 
  $(\chi\cdot a)(a')=\chi(aa')$ 
  for  $a, a' \in A$ and  $\chi \in \widehat{A}$. Furthermore, $A$ is a Frobenius ring if and only if there is an isomorphism of right $A$-modules between $A$ and $\widehat{A}$ (see for instance \cite[Theorem 3.10]{Wood99}). 
\begin{theo}
The ring $R$ is a Frobenius ring.
\end{theo}
\begin{proof}
Since $K$ is a Frobenius ring, there is $\chi_K \in \widehat{K}$ such that
$$
\Phi :x \in K \mapsto (\chi_K \cdot x) \in \widehat{K}
$$ is an isomorphism of right $K$-modules. We define $\chi_R \in \widehat{R}$ by $\chi_R\left(\sum_g a_g \overline{g}\right)=\chi_K(a_1)$ and claim that the following morphism of right $R$-modules 
$$
\Psi: r \in R \mapsto \chi_R \cdot r \in \widehat{R}
$$ is an isomorphism. Since $| R |=| \widehat{R}|$, it remains to prove that $\Psi$ is injective. For that, let $r=\sum_g r_g \overline{g} \in R$ and $r'=\sum_g r'_g \overline{g} \in R$ such that $\Psi(r)=\Psi(r')$. For $g \in G$ and $x \in K$, we have  
$$
\left(\chi_R \cdot r\right)\left( (\Theta\left( g^{-1}\right)x)\overline{g^{-1}}\right)=\left(\chi_R \cdot r'\right)\left( (\Theta\left(g^{-1}\right)x)\overline{g^{-1}}\right),
$$
that is,
$$
\chi_R\left(r \cdot \left( (\Theta(g^{-1})x)\overline{g^{-1}}\right)\right)=\chi_R\left(r' \cdot \left( (\Theta(g^{-1})x)\overline{g^{-1}}\right)\right),
$$ and 
so,
$$\chi_K\left(r_g \alpha(g,g^{-1})x\right)=\chi_K\left(r'_g \alpha(g,g^{-1}x)\right).$$
Since 
$$
\left(\chi_K \cdot \left(r_g \alpha(g,g^{-1})\right)\right)(x)=\chi_K\left(r_g \alpha(g,g^{-1})x\right),$$
we obtain $\chi_K \cdot \left(r_g \alpha(g,g^{-1})\right)=\chi_K \cdot \left(r'_g \alpha(g,g^{-1})\right)$. As $\Phi$ is bijective, we finally get $r_g \alpha(g,g^{-1})=r'_g \alpha(g,g^{-1})$, and so $r_g=r'_{g}$ for all $g \in G$, which proves $r=r'$. Therefore, the proof is complete.
\end{proof}

Next we prove an extension of Maschke's Theorem, which is well-known for twisted group algebras (see \cite[Chap. 3, Theorem 2.10]{K85}).

\begin{theo} \label{Maschke} If $\cha K \nmid |G|$, then $R$ is a semisimple ring.
\end{theo}
\begin{proof} Let $V$ be an $R$-left module and $W$ a submodule of $V$. We may write $V= W \oplus U$ with a $K$-vector space U.
Let $\rho$ be the projection of $V$ onto $W$ with kernel $U$. For $v \in V$ we define the map $\Lambda: V \longrightarrow W$ by
$$ \Lambda(v) = \frac{1}{|G|} \sum_{g \in G} \frac{1}{\alpha(g,g^{-1})}\overline{g} \rho(\overline{g^{-1}}v) 
= \frac{1}{|G|} \sum_{g \in G}\overline{g} \rho(\overline{g}^{-1}v). $$
If $a \in K$, then

$$
\begin{array}{rcl}
\Lambda(av) & = & \frac{1}{|G|} \sum_{g \in G}  \overline{g} \rho(\overline{g}^{-1}av) \\[1ex]
&=& \frac{1}{|G|} \sum_{g \in G}  \overline{g} \rho((\Theta(g^{-1})a)\overline{g}^{-1}v) \\[1ex]
 & = & \frac{1}{|G|} \sum_{g \in G}  \overline{g} (\Theta(g^{-1})a)\rho(\overline{g}^{-1}v) \\[1ex] 
 & = &\frac{1}{|G|} \sum_{g \in G} 
\Theta(g)(\Theta(g^{-1})a) \overline{g}\rho(\overline{g}^{-1}v)\\[1ex]
 &=& \frac{1}{|G|} \sum_{g \in G}  a \overline{g}\rho(\overline{g}^{-1}v) = a \Lambda(v).
\end{array} 
$$
Thus $\Lambda$ is $K$-linear.
For  $h \in G$, we obtain
$$ \begin{array}{rcl}
\Lambda(\overline{h}v) & = & \frac{1}{|G|} \sum_{g \in G}\overline{g} \rho(\overline{g}^{-1}\,\overline{h}v)
= \frac{1}{|G|} \sum_{x \in G}\overline{hx} \rho({(\overline{hx})^{-1}}\,\overline{h}v).
\end{array}
$$
Since $\overline{hx}^{-1} = \alpha(h,x) \overline{x}^{-1}\overline{h}^{-1}$ and $\Theta(g)\alpha(x,y) =\alpha(x,y)$, we get
$$ \begin{array}{rcl}
\Lambda(\overline{h}v) & = & 
\frac{1}{|G|} \sum_{x \in G} \alpha(h,x)^{-1}\overline{h}\overline{x} \rho((\alpha(h,x)\overline{x}^{-1}\overline{h}^{-1}\overline{h}v) \\[1ex]
&=& \overline{h}\frac{1}{|G|} \sum_{x \in G} \alpha(h,x)^{-1}\overline{x} \rho((\alpha(h,x)\overline{x}^{-1}\overline{h}^{-1}\overline{h}v) \\[1ex]
&=& \overline{h}\frac{1}{|G|} \sum_{x \in G} \overline{x} \rho(\overline{x}^{-1}v) \\[1ex]
&=& \overline{h}\Lambda(v).
\end{array}
$$
Thus $\Lambda$ is $R$-linear and the identity on $W$. It follows that $V = W \oplus \kernel\Lambda$ where $\kernel\Lambda$ is a
$R$-left module.
\end{proof}

\begin{defi} For  $a =\sum_{g \in G} a_g \overline{g} \in R$ we define the adjoint $\widehat{a}$ of $a$ by
$$ \widehat{a} = \sum_{g \in G} (\Theta(g^{-1})a_g)\alpha(g,g^{-1}) \overline{g^{-1}}.$$
\end{defi}

Note that the map \, $\widehat{}$ \ is not K-linear. If we consider $\widehat{a} \in K[G,\Theta,\alpha^{-1}]$ we have the following.

\begin{theo}\label{antiiso} \  $\widehat{}\ : R=K[G,\Theta,\alpha] \longrightarrow K[G,\Theta,\alpha^{-1}]$ is  distance preserving ring anti-isomorphism.
Moreover, if $\alpha = \alpha^{-1}$, then \ $\widehat{\widehat{a}} = a$ for all $a \in R$.
\end{theo}
\begin{proof} Clearly, the map \ $\widehat{}$ \ is additive and maps the identity  onto the identity since $\alpha$ is normalized.
Furthermore, \ $\widehat{}$ \ defines a bijection. Thus, we only have to prove $\widehat{a\,b} = \widehat{b}\,\widehat{a}$, for
$a,b \in R$. Let $a = \sum_{g \in G} a_g \overline{g}$ and $b = \sum_{h \in G} b_h \overline{h}$.
One easily computes 
\begin{align*}
\widehat{a\, b} = \sum_{g,h \in G} (\Theta(h^{-1}g^{-1})a_g)(\Theta(h^{-1}b_h)\alpha(g,h)\alpha(gh,h^{-1}g^{-1}) \overline{h^{-1}g^{-1}}
\end{align*}
and
\begin{align*}
\widehat{b}\, \widehat{a} = \sum_{g,h \in G} (\Theta(h^{-1}g^{-1})a_g)(\Theta(h^{-1}b_h)\alpha(h,h^{-1})\alpha(g,g^{-1})
\alpha^{-1}(h^{-1},g^{-1})\overline{h^{-1}g^{-1}}. 
\end{align*}
Applying the normalized cocycle condition, we see that
$$ \alpha(g,h)\alpha(gh,h^{-1}g^{-1}) = \alpha(h,h^{-1})\alpha(g,g^{-1})\alpha^{-1}(h^{-1},g^{-1}),$$
hence   $\widehat{a\,b} = \widehat{b}\,\widehat{a}$.
Finally note that
$$ \widehat{\widehat{a}} = \sum_{g \in G} \Theta(g)\left[(\Theta(g^{-1})a_g)\alpha(g,g^{-1})\alpha(g^{-1},g)\right]\overline{g} =
 \sum_{g \in G} a_g\alpha(g,g^{-1})\alpha(g^{-1},g)\overline{g}, $$
which proves $\widehat{\widehat{a}} = a$, since $\alpha(g,g^{-1})=\alpha(g^{-1},g)$ which follows from the cocycle condition,
and $\alpha = \alpha^{-1}$, by assumption. 
\end{proof}

Let $K^\alpha G= K[G,1,\alpha]$ be a twisted group code over a field $K$ of characteristic $p$. The set
$ \widehat{G} = K^* \times G$ becomes a group via the multiplication
$$(\lambda,g)\cdot (\mu,h) = (\alpha(g,h)\lambda \mu,gh).$$ Since $\alpha(1,g)=\alpha(g,1)$ for all $g \in G$, we see that
$$N=\{(\lambda,1) \mid \lambda \in K^*\} \cong K^*$$ is a central subgroup of $\widehat{G}$ and
$G \cong \widehat{G}/N$.

\begin{prop} \label{P1} Let $\cha K =p$. If $G$ is $p$-nilpotent with a cyclic Sylow $p$-subgroup, then every left (right) ideal in a twisted group algebra is principal.
\end{prop}
\begin{proof} By \cite[Proposition 1.2.17]{L18}, the $K$-linear map $\rho: K\widehat{G} \longrightarrow K^\alpha G$ given by
$$ (\lambda,g) \mapsto \lambda \overline{g}$$
defines an $K$-algebra epimorphism. Note that
$\widehat{G}$ is $p$-nilpotent with cyclic Sylow $p$-subgroups. Now let $\overline{C}$ be a left ideal in $K^\alpha G$ and $C$ its preimage in $K\widehat{G}$. By \cite[Corollary 3.2]{BCW22}, $C$ is a principal ideal, i.e., $C=aK\widehat{G}$. Thus
$\overline{C} = \rho(a)K^\alpha G.$

\end{proof}

\begin{question} {\rm Is there an extension of Proposition \ref{P1} to left ideals in  twisted skew group rings?}
\end{question}

\section{Twisted skew group codes}


We shall use $R$ with the natural basis $\{\overline{g} \mid g \in G \}$ as an ambient space. Thus $R$ carries the Hamming distance and the
Euclidean bilinear form defined by  $\langle \overline{g}, \overline{h} \rangle = \delta_{g,h}$.

\begin{defi} \label{D1} A left ideal $C$ in $R=K[G,\Theta,\alpha]$ is called a {\it twisted skew group code}, more precisely a
{\it $(G,\Theta,\alpha)$-code}. We shortly write $C \leq R$ in this case.
\end{defi}
Note that $(G,1,1)$-codes are nothing else than group codes \cite{W21}, and $(G,1,\alpha)$-codes are twisted group codes \cite{CW21}. Inspired by \cite{BGU07} (see also \cite{D21}) we call $(G,\Theta,1)$-codes {skew group codes}. \\

For the reader's convenience we repeat a well-known definition.

\begin{defi} Let $\theta \in \Aut(K)$
and $\lambda \in K^*$. A linear code $C$ in $K^n$ is called \\
a) {\it $\theta$-cyclic} (or {\it skew-cyclic}) if $$(c_0,c_1, \ldots,c_{n-1}) \in C \Longrightarrow
 (\theta(c_{n-1}), \theta(c_0), \ldots, \theta(c_{n-2})) \in C,$$ \\
b) {\it $(\theta,\lambda)$-constacyclic}
(or {\it skew $\lambda$-constacyclic}) if $$(c_0,c_1, \ldots,c_{n-1}) \in C \Longrightarrow
 (\lambda \theta(c_{n-1}), \theta(c_0), \ldots, \theta(c_{n-2})) \in C.$$

\end{defi}

Let $C_n =\langle g \rangle$ and let $\theta \in \Aut(K)$ with $\ord(\theta) \mid n$.
We define $\Theta$ by $\Theta(g^i) =\theta^i$ for $i=0,\ldots, n-1$, and for $\lambda \in K^*$, the cocycle
$\alpha_\lambda$ by
\begin{equation} \label{consta-cocycle} \alpha_\lambda(g^i,g^j) =\left\{ \begin{array}{cl} 
  1 & \text{if} \ 0 \leq i+j < n \\ \lambda & \text{if} \ n\leq i+j \leq 2(n -1).
\end{array} \right.
\end{equation}
With this notation $(C_n,\Theta,1)$-codes are $\theta$-cyclic codes \cite{BGU07}, and $(C_n,\Theta,\alpha_\lambda)$-codes are
skew $\lambda$-constacyclic codes if $\theta(\lambda)=\lambda$ \cite{BBB19}. \\

Thus with the Definition \ref{D1} we cover many classes
of interesting codes in twisted skew group rings by specializing parameters. \\

The following result shows that all twisted skew group codes for a cyclic group are skew $\lambda$-constacyclic.

\begin{prop}
If $R=K[C_n,\Theta,\alpha]$ is a twisted skew group ring, then $R$ is skew $\lambda$-constacyclic for some $\lambda \in K^*$ fixed by $\Theta(g)$.
\end{prop}
\begin{proof}
Let $k$ denote the subfield of $K$ which is fixed by $\Theta(g)$. By the assumption on $\alpha$ (see (\ref{eq1})), the cocycle $\alpha$ takes only values in $k$. Thus, by
\cite[Problem 16, Chapter 10]{Mor96} there exists  $\lambda = \prod_{i=0}^{n-1}\alpha(g^i,g) \in k^*$ and  $\kappa: G \rightarrow k^*$ with $\kappa(1)=1$ such that
$$ \alpha(g^i,g^j) = \kappa(g^i) \kappa(g^j) \kappa(g^{i+j})^{-1} \alpha_\lambda(g^i,g^j).$$
 Using Lemma \ref{lem:isogroupcod}, with $\delta=id_{G}$, we have $K[G,\Theta,\alpha]\simeq K[G,\Theta,\alpha_\lambda]$.
\end{proof}

The next result was first proved for group codes in \cite{BWZ22}. An extension to twisted
group codes has been given in \cite{w23}. For the reader's convenience we repeat the proof here adapted to twisted skew group codes.

\begin{theo} \label{bound} If $0 \not= C \leq K[G,\Theta,\alpha] =R$ is a twisted skew group code of minimum distance
$\dist(C)$, then
$$ |G| \leq \dist(C) \cdot \dim C.$$
\end{theo}

\begin{proof}
 We obviously may replace $C$ by $C = Rc$ with $\wt(c) = \dist(C)$. Let $c = \sum_{g \in G} c_g\overline{g}$
and put $S = \supp(c) =\{g \mid c_g \not= 0\} \subseteq G$. Let $g_1, \ldots,g_t$ be a maximal set of pairwise different elements of $G$ such that
$$ g_iS \not\subseteq \cup_{j<i} \,  g_jS$$ for all $i=1,\ldots,t$. By maximality, we have $G= \cup_{i=1}^t g_iS$, hence $t \geq\frac{|G|}{|S|}$.
Next we consider the map
$T_c: R \longrightarrow R$ defined by $a \mapsto T_c(a) = ac$. Note that $T_c$ is $K$-linear from the left.
Since $\rank T_c = \dim C$, it is sufficient to prove that
for $i=1,\ldots,t$ the vectors $T_c(\overline{g_i})$ are $K$-linear independent. Suppose that
$v=\sum_{i=1}^t \lambda_i T_c(\overline{g_i}) = 0$ with $\lambda_i\in K$. Let $r =\max\{i \mid \lambda_i \not= 0\}$.
If
$$ h = g_rs \in g_rS \setminus \cup_{j<r} g_jS $$ with $s \in S,$
then the coefficient of $\overline{h}$ in $v$ is equal to
$$ \lambda_r(\Theta(g_r) c_s)\alpha(g_r,s) \not= 0, $$
a contradiction.
\end{proof}

Group codes for which equality holds in the bound in Theorem
\ref{bound} have been characterized completely in \cite{BWZ22}.

\section{Equivalence and automorphism groups}

Let $M(n,K)$ (resp., ${\rm Diag}(n,K)$) be the set of monomial (resp., diagonal) matrices  of type $n \times n$ over $K$. The group ${\rm Aut}(K)$ acts naturally from the left on $M(n,K)$ by  $\gamma \cdot A=(\gamma(a_{i,j}))$ for $\gamma \in {\rm Aut}(K)$ and $A=(a_{i,j}) \in M(n,K)$.  In the sequel, for convenience, we denote $\gamma \cdot A$ by $A^{\gamma^{-1}}$. In this way, we get the semidirect product $\Gamma {\rm L}_n= {\rm Aut}(K) \ltimes M(n,K)$ with multiplication 
$$
(\gamma,A)\cdot (\beta,B)=(\gamma\beta, A^{\beta}B)
$$
for $(\gamma,A),(\beta,B) \in \Gamma {\rm L}_n$.
In this section, we regard elements of $K^n$ as column vectors. The following defines a left action of $\Gamma {\rm L}_n$ on $(K^n,+)$. For $x=(x_1,\dots,x_n)^T \in K^n$ and $(\gamma,A) \in M(n,K)$, we put
$$
(\gamma,A) \cdot x=A^{\gamma^{-1}}x^{\gamma^{-1}},
$$ where $x^{\gamma^{-1}}=(\gamma(x_1),\gamma(x_2),\dots,\gamma(x_n))^T$. 
{ Observe that this map is semilinear since, for $k \in K$, we have
$$ (\gamma,A)\cdot (kx) = \gamma(k)\cdot ((\gamma,A)x). $$}
From now on, let $\Pi_1$ (resp., $\Pi_2$) be the projection of $\Gamma {\rm L}_n$ on ${\rm Aut}(K)$ (resp., $M(n,K)$).

\begin{lem}\label{lem:embSn}
The action of $\Gamma {\rm L}_n$ on $(K^n,+)$ is faithful and (Hamming-)distance preserving. Hence we may identify $\Gamma {\rm L}_n$ as a subgroup of ${\rm Aut}(K^n,+)$.
\end{lem} 
\begin{proof}
The proof is classical. Consider the standard $K$-basis $\{e_1, \dots,e_n\}$ of $K^n$. Let $(\gamma,A) \in \Gamma {\rm L}_n$ such that $(\gamma,A)x=x$, for all $x \in K^n$. Taking $x=e_j$ ($j=1,\dots,n$), we deduce that $A^{\gamma^{-1}}e_j=e_j$, and so $A^\gamma={\rm I}_n$, or equivalently $A={\rm I}_n$.  Therefore, $x^{\gamma^{-1}}=x$ for all $x \in K^n$, and so $\gamma={\rm Id}_K$.
\end{proof}

\begin{lem}\label{lem:semilinear}
 $\Gamma {\rm L}_n$  is the group of semilinear isometries of the Hamming space $(K^n,+,d)$.
\end{lem} 
\begin{proof} Let $f$ be a semilinear
bijection of $K^n$, i.e.,
$$f(x+y) = f(x) +f(y) \quad \text{and} \quad f(kx)=\gamma(k)f(x) $$ for $x \in K^n$ and $k \in K$ where $\gamma \in \Aut(K)$. Assume also that $f$ is an isometry of $(K^n,+,d)$. We only have to show that $f \in \Gamma {\rm L}_n.$
If $e_1, \ldots, e_n$ is the standard basis of $K^n$, then
$$ f(e_i) = \gamma(k_i)e_{\pi(i)}$$
for suitable $k_i\in K^*$, where $\pi$ is a permutation of $\{1,\ldots,n\}$. Thus,
for $x_i\in K$, we get
$$f\left(\sum_{i=1}^n x_i e_i\right) = \sum_{i=1}^n \gamma(x_i)f(e_i) = \sum_{i=1}^n \gamma(x_ik_i) e_{\pi(i)} = \sum_{i=1}^n \gamma(x_{\pi^{-1}(i)}k_{\pi^{-1}(i)}) e_i.$$ This obviously describes a map
in $\Gamma {\rm L}_n$. 
\end{proof}

Recall that for left ideals $C \leq  R=K[G,\Theta,\alpha]$ the left multiplication of $g \in G$  is given by
$$ g\cdot(kc) = (\Theta(g)k)(g\cdot c)$$
for $k\in K$ and $c \in C$, hence by a semilinear map.
Thus $C$ is invariant under a group of suitable semilinear maps and we may define the {\it automorphism group} of $C$ as
 
$$
{\rm Aut}(C)=\{ (\gamma,A) \in \Gamma {\rm L}_n \mid (\gamma,A)\cdot C =C  \},
$$ which is clearly a subgroup of $\Gamma {\rm L}_n$. Note that a monomial matrix $A \in M(n,K)$ over $K$ is of the form
$$
A={\rm diag}(a_1(A),\dots,a_n(A))P(A), 
$$ where ${\rm diag}(a_1(A),\dots,a_n(A))$ is the diagonal matrix with entries $a_i(A) \in K^*$ and $P(A)$ is the permutation matrix of the permutation induced by $A$ on $\{1,\dots,n\}$. We call the diagonal matrix $D(A)={\rm diag}(a_1(A),\dots,a_n(A))$ the {\it diagonal part} and $P(A)$ the {\it permutation part} of $A$.

\begin{defi} Let $C \leq K[G,\Theta,\alpha]$ and $C' \leq K[G',\Theta',\alpha']$ with $|G|=|G'|=n$. Then we call $C$ and $C'$
{\it equivalent} if there exists $(\gamma,A) \in {\Gamma {\rm L}}_n$ such that 
$(\gamma,A)C = C'.$ 
\end{defi}

\begin{rk} Note that the map
$K[G,\Theta,\alpha] \ni a \mapsto 
\widehat{a} \in K[G,\Theta,\alpha^{-1}]$ can be described by a suitable $(\gamma,A) \in {\Gamma {\rm L}}_{|G|}$, but it maps a left ideal onto a right ideal (see Theorem \ref{antiiso}).
\end{rk}

Now let
$$
G(\Theta,\alpha)=K^* \times G.
$$ We define a group structure on $G(\Theta,\alpha)$ by:
$$
(\lambda,g)\cdot (\mu,h)=(\lambda\Theta(g)(\mu)\alpha(g,h),gh),
$$
for all $(\lambda,g),(\mu,h) \in G(\Theta,\alpha)$. Note that
$$
N=\{(\lambda,1) \mid \lambda \in K^* \} \simeq K^*
$$ is a normal subgroup of $G(\Theta,\alpha)$ and $G(\Theta,\alpha)/N \simeq G$.
On $(R,+)$, the group $G(\Theta,\alpha)$ acts faithfully in the following way: for $(\lambda,g) \in G(\Theta,\alpha)$, $\eta \in K$ and $h \in G$, we define
$$
(\lambda,g)\cdot(\eta\overline{h})=\lambda\Theta(g)(\eta)\alpha(g,h)\overline{gh}.
$$ 
With respect to the $K$-basis $\{ \overline{g}\}_{g \in G}$ of $R=K[G,\Theta,\alpha]$, we may identify $G(\Theta,\alpha)$ with a subgroup in 
$\Gamma{\rm L}_n$.

From now on, if $u \in \Gamma{\rm L}_n$, we let $P(u)=P(\Pi_2(u))$, the permutation part of $\Pi_2(u)$. 

\begin{theo}
Let $C \leq K^n$ be a linear code. Then $C$ is a twisted skew group code in $K[G,\Theta,\alpha]$ for a suitable $\Theta$ and a cocycle $\alpha$ of $G$ if and only if there exists a subgroup $\widetilde{G} \leq {\rm Aut}(C)$ satisfying the following three conditions.
\begin{itemize}
\item[\rm a)] The only diagonal matrices in $\widetilde{G} \cap M(n,K)$ are scalar matrices.
\item[\rm b)] $G=P(\widetilde{G})$ acts regularly on $\{1,\dots,n\}$.
\item[\rm c)] If $u,v \in \widetilde{G}$ such that $P(u)=P(v)$, then $\Pi_1(u)=\Pi_1(v)$.
\end{itemize}
\end{theo}
\begin{proof} To prove the if part we put
$\widetilde{G}=\xi(G(\Theta,\alpha))$.
Let $c=\sum_h c_h \overline{h} \in C$ and $(\lambda,g) \in G(\Theta,\alpha)$. Then
\begin{align*}
(\lambda,g)\cdot c&=\sum_{h \in G} \lambda \Theta(g)(c_h)\alpha(g,h)\overline{gh}\\
&= (\lambda \overline{g}) \cdot \left( \sum_{h \in G} c_h \overline{h} \right) =(\lambda \overline{g})\cdot c   \in C,
\end{align*} since $C$ is a left ideal in $K[G,\Theta,\alpha]$. Thus, we get $\widetilde{G}=\xi(G(\Theta,\alpha)) \leq {\rm Aut}(C)$ and the three conditions are
easily checked.

To prove the only if part, let $\widetilde{G} \leq {\rm Aut}(C)$ as in the statement. Let $B=\{e_i \mid i=1, \dots, n \}$ be the standard basis of $K^n$. Since $G$ acts regularly on $\{1,\dots,n\}$ we may identify $i \in \{1,\dots,n\}$ with the unique element $h \in G$ such that $h \cdot 1 =i $. Thus we can write $B=\{e_g \mid g \in G\}$. For any $g \in G$, note that $g \cdot i = g\cdot( h \cdot 1)=(gh)\cdot 1$ and so $g \cdot i$ corresponds to $gh$. Furthermore note that if $\widetilde{g} \in \widetilde{G}$ such that $P(\widetilde{g})=g$, then there exists $\beta \in K^*$ such that $\widetilde{g}\cdot e_1= \beta e_g$. 

{\bf Claim:} For any $g \in G$, there exists a unique $\widetilde{g} \in \widetilde{G}$ such that $P(\widetilde{g})=g$ and $\widetilde{g}e_1=e_g$. Indeed, let $\widetilde{g}_1,\widetilde{g}_2 \in \widetilde{G}$ such that $P(\widetilde{g}_1)=P(\widetilde{g}_2)=g$. By condition c), we may write $\widetilde{g}_1=(\gamma,D_1\cdot g)$ and $\widetilde{g}_2=(\gamma,D_2\cdot g)$ for some $\gamma \in {\rm Aut}(K)$ and some diagonal matrices $D_1,D_2$. Thus we have
$$
\widetilde{g}_1\widetilde{g}_2^{-1}=(\gamma,D_1\cdot g)(\gamma^{-1},(D_2^{\gamma^{-1}} \cdot g^{\gamma^{-1}})^{-1})=({\rm Id}_K,D_1^{\gamma^{-1}}(D_{2}^{\gamma^{-1}})^{-1}) \in \widetilde{G}.
$$ By condition a),
 we get $\widetilde{g_2}=({\rm Id}_K,\lambda{\rm I}_n)\widetilde{g}_1$ for some $\lambda \in K^*$.  Thus the claim follows. In the sequel, we denote by $\Theta(g)$ the unique element $\gamma$ as defined above.

Now let $g,h \in G$. Since $P(\widetilde{g}\widetilde{h})=P(\widetilde{gh})=gh$, there exists a unique $\alpha(g,h) \in K^*$ such that 
\begin{equation}\label{eq:relfjf}
\widetilde{g}\widetilde{h}=({\rm Id}_K,\alpha(g,h){\rm I}_n)\widetilde{gh},
\end{equation} by the proof of the Claim.

Moreover, $\widetilde{g}e_h=\alpha(g,h)e_{gh}$. Indeed, we have 
\begin{align*}
\widetilde{g}e_h=\widetilde{g}(\widetilde{h}e_1)=(\widetilde{g}\widetilde{h})e_1
&=({\rm Id}_K,\alpha(g,h){\rm I}_n)\widetilde{gh}e_1 \quad (\text{using}\,\, \eqref{eq:relfjf})\\
&=({\rm Id}_K,\alpha(g,h){\rm I}_n)e_{gh}=\alpha(g,h)e_{gh}.
\end{align*}

From now on, let $\Theta: g \in G \mapsto \Theta(g) \in {\rm Aut}(K)$, where $\Theta(g)$ is defined as in the proof of teh Claim. First, note that for all $g,h \in G$ and $\lambda \in K$, we have 
$$\widetilde{g}\cdot(\lambda e_h)=\Theta(g)(\lambda)(\widetilde{g}e_h).\footnote{Indeed, there exists a diagonal matrix $D$ such that $\widetilde{g}=(\Theta(g),D\cdot g)$. Then we have $\widetilde{g}(\lambda e_h)= D^{\Theta(g^{-1})} \cdot g^{\Theta(g^{-1})}(\lambda e_h)^{\Theta(g^{-1})}=\lambda^{\Theta(g^{-1})}D^{\Theta(g^{-1})} \cdot g^{\Theta(g^{-1})}(e_h)^{\Theta(g^{-1})}=\Theta(g)(\lambda)(\widetilde{g}e_h)$.}
$$ Next we claim that $\Theta$ is a morphism of groups. For this purpose, let $g,h \in G$ and $\lambda \in K$. We have
\begin{align*}
\widetilde{g}\cdot(\widetilde{h}\cdot( \lambda  e_1))=\widetilde{g}\cdot (\Theta(h)(\lambda)e_{h})&=(\Theta(g)\circ \Theta(h))(\lambda)\widetilde{g}e_h\\
&=(\Theta(g)\circ \Theta(h))(\lambda)\alpha(g,h)e_{gh}
\end{align*}
and
\begin{align*}
\widetilde{g}\cdot (\widetilde{h} \cdot (\lambda e_1))=(\widetilde{g} \cdot 
\widetilde{h})\cdot (\lambda e_1)&=(({\rm Id}_K,\alpha(g,h){\rm I}_n)\widetilde{gh}) \cdot (\lambda e_1)\qquad \text{(using  \eqref{eq:relfjf} again})\\
&=({\rm Id}_K,\alpha(g,h){\rm I}_n)(\Theta(gh)(\lambda)(\widetilde{gh}e_1))
\\
&=\alpha(g,h)\Theta(gh)(\lambda)e_{gh}.
\end{align*}
Thus $\Theta(g) \circ \Theta(h)=\Theta(gh)$. Since $P(({\rm Id}_K,{\rm I}_n))=1$ and $({\rm Id}_K,{\rm I}_n)e_1=e_1$, we conclude, by applying the Claim, that $\widetilde{1}=({\rm Id}_K,{\rm I}_n)$ and so $\Theta(1)={\rm Id}_K$. Consequently, $\Theta$ is a morphism of groups.

Finally  we show that $\alpha$ is a cocycle, that is,
$$
\alpha(gh,l)\alpha(g,h)=\Theta(g)(\alpha(h,l))\alpha(g,hl),
$$ for all $g,h,l \in G$.
This follows from
$$
(\widetilde{g} \widetilde{h}) e_l=\widetilde{g}(\widetilde{h}e_l)=\widetilde{g}(\alpha(h,l)e_{hl})=\Theta(g)(\alpha(h,l))(\widetilde{g}e_{hl})=\Theta(g)(\alpha(h,l))\alpha(g,hl)e_{ghl}.
$$
and
$$(\widetilde{g} \widetilde{h}) e_l=(({\rm Id}_K,\alpha(g,h){\rm I}_n)\widetilde{gh})e_l=({\rm Id}_K,\alpha(g,h){\rm I}_n)(\alpha(gh,l)e_{ghl})= \alpha(g,h)\alpha(gh,l)e_{ghl}.$$

Furthermore $\alpha$ is normalized. To prove this let $g \in G$. Since $\widetilde{1}=({\rm Id}_K,{\rm I}_n)$ we get 
$$
e_g=\widetilde{1}e_g=\alpha(1,g)e_g
$$
and
$$
e_g=\widetilde{g}e_1=\alpha(g,1)e_g,
$$ and so $\alpha(1,g)=\alpha(g,1)=1$.

To the final end we consider the $K$-linear isomorphism
$$
\Psi : K^n \rightarrow K[G,\Theta,\alpha]
$$ defined by
$\Psi(e_g)=\overline{g}$. Let $c=\sum_{h}c_h e_h \in C$, $\beta \in K$ and $g \in G$. We have
\begin{align*}
\beta\overline{g}\Psi(c)&=\beta \overline{g}\left( \sum_h c_h \overline{h}\right)\\
&=\beta\sum_h\Theta(g)(c_h)\alpha(g,h)\overline{gh}\\
&=\Psi\left(\beta\sum_h\Theta(g)(c_h)\alpha(g,h)e_{gh}\right)\\
&=\Psi\left(({\rm Id}_K,\beta{\rm I}_n)\widetilde{g}\cdot c\right) \in \Psi(C)\quad \text{since}\,\, \widetilde{G}\cdot C\leq C.
\end{align*} Hence $\Psi(C)$ is a left ideal of $K[G,\Theta,\alpha]$, which completes the proof.
\end{proof}

\begin{coro}
Let $C \leq K^n$ be a linear code. Then $C$ is a skew group code in $K[G,\Theta,1]$ for a suitable $\Theta$ if and only if there exists a subgroup $\widetilde{G} \leq {\rm Aut}(C)$ satisfying the following two conditions.
\begin{itemize}
\item[\rm a)] $G=P(\widetilde{G})$ acts regularly on $\{1,\dots,n\}$.
\item[\rm b)] If $u,v \in \widetilde{G}$ such that $P(u)=P(v)$, then $\Pi_1(u)=\Pi_1(v)$.
\end{itemize}
\end{coro}

\section{Duality}

As usual we denote by $C^\perp$ the dual of a code $C$.
The first result  is an extension of a Theorem of MacWilliams \cite{MacW}.

\begin{theo} \label{eucl} If $C \leq R$ is a left ideal in $R$, then 
$ C^\perp = \widehat{\ann_r(C)}$, where the right annihilator $\ann_r(C)$ of $C$ is defined as   $\ann_r(C) = \{ a \in R \mid c\cdot a = 0 \ \text{for all} \ c \in C \}$.
\end{theo}
\begin{proof} If $c = \sum_{h \in G} c_h \overline{h} \in C$ and $a=\sum_{g \in G}a_g \overline{g} \in \ann_r(C)$, then
$$ 0 = ca=\sum_{g,h \in G} c_h(\Theta(h)a_g)\alpha(h,g)\overline{hg}. $$
In particular,
$$0 = \sum_{g \in G} c_{g^{-1}}(\Theta(g^{-1})a_g)\alpha(g^{-1},g) = \sum_{g \in G} c_{g^{-1}}(\Theta(g^{-1})a_g)\alpha(g,g^{-1}), $$
since $\alpha(g,g^{-1})= \alpha(g^{-1},g),$ which is a consequence of the cocycle condition
of $\alpha$. It follows that
$$ \langle c, \widehat{a} \rangle = \sum_{g \in G} c_{g^{-1}}(\Theta(g^{-1})a_g)\alpha(g,g^{-1}) =0,$$
hence $\widehat{\ann_r(C)} \subseteq C^\perp$. Next we consider the map
$$ \Lambda: R \longrightarrow C^* = \ho_K(C,K) $$
where $\Lambda(a)c = (ca)_{\overline{1}},$ the scalar in $ca$ at $\overline{1}$. Note that $\Lambda$ is not $K$-linear, but
a group homomorphism from $(R,+)$ into $(C^*,+)$. Suppose that $ 0 = (ca)_{\overline{1}}$ for all $c \in C$ and
 $ca = \sum_{g \in G}k_g \overline{g}$ with $k_x \not= 0$ for some $1 \not=x \in G$. Since $C$ is a left ideal, we have $\overline{x^{-1}}c \in C$
and $(\overline{x^{-1}}c)a$ has the scalar $$(\Theta(x^{-1})k_x)\alpha(x^{-1},x)\not=0$$ at $\overline{1}$, a contradiction.
This shows that $\kernel \Lambda = \ann_r(C)$. It follows 
$$\begin{array}{rcl}
|C^*| & \geq & \frac{|R|}{|\kernel \lambda|} = \frac{|R|}{|\ann_r(C)|} =  \frac{|R|}{|\widehat{\ann_r(C)}|} \geq \frac{|R|}{|C^\perp|} =
|C| = |\ho_K(C,K)|= |C^*|,
\end{array}
$$
hence $ C^\perp = \widehat{\ann_r(C)}$.
\end{proof}

\begin{rk} For group codes $C$ the dual $C^\perp$ is a group code as well. This is far away to be true in general for twisted group codes. Note, that according to Theorem \ref{antiiso} and Theorem \ref{eucl},
$C^\perp$ is only a left ideal in $K[G,1,\alpha^{-1}]$
if $C \leq K[G,1,\alpha]$.  This fact causes many problems when dealing with $C^\perp$ in the case $\alpha \not= \alpha^{-1}$. 

\end{rk}

\begin{defi}
Let $|K|=q^2$ and let $C \leq R.$ \\
a) The Hermitian  form $\langle \cdot\, , \cdot\rangle_h $  on $R$ is defined by
$\langle a, b\rangle_h = \sum_{g \in G} a_g(b_g)^q = \sum_{g \in G} a_g b_g^q$. \\
b) We denote the dual of $C \leq R$ with respect to the 
Hermitian form  by $C^{\perp_h}$. \\
c) For $a = \sum_{g \in G} a_g \overline{g}$ we put
$a^{(q)} = \sum_{g \in G} a_g^q \overline{g}$. \\
d) $C^{(q)} = \{ c^{(q)} \mid c \in C \}$.
\end{defi}

\begin{theo} \label{herm} Let $|K|=q^2$ and let $ C \leq R$ be a left ideal in $R$. If $\alpha(x,y)^q =\alpha(x,y)$ for all $x,y \in G$, then 
$ C^{\perp_h} =\widehat{\ann_r(C^{(q)})}$.
\end{theo}
\begin{proof} 
If $c = \sum_{h \in G} c_h \overline{h} \in C$ and $a=\sum_{g \in G}a_g \overline{g} \in \ann_r(C^{(q)})$, then
$$ 0 = c^{(q)}a=\sum_{g,h \in G} c^q_h(\Theta(h)a_g)\alpha(h,g)\overline{hg}. $$
In particular,
$$0 = \sum_{g \in G} c^q_{g^{-1}}(\Theta(g^{-1})a_g)\alpha(g^{-1},g) = \sum_{g \in G} c_{g^{-1}}(\Theta(g^{-1})a_g)^q\alpha(g,g^{-1}), $$
since $\alpha(g,g^{-1})= \alpha(g^{-1},g).$ It follows
$$ \langle c, \widehat{a} \rangle_{h} = \sum_{g \in G} c_{g^{-1}}(\Theta(g^{-1})a_g)^q\alpha(g,g^{-1}) =0,$$
hence $\widehat{\ann_r(C^{(q)})} \subseteq C^{\perp_h}$. Note that $C^{(q)}$ is a left ideal in $R$ since $\alpha$ is fixed by
the Frobenius automorphism $x \mapsto x^q$. The equality $\widehat{\ann_r(C^{(q)})} = C^\perp_h$ now follows as
in the Euclidean case.
\end{proof}



\section{Idempotent codes} \label{idempotents}

\begin{defi} A code $C \leq R$ is {\it idempotent} if there exists an idempotent $e \in R$, i.e., $e=e^2$ such that $C=Re$.
\end{defi}



\begin{defi} Let $C\leq R$ a twisted skew group code. 
\begin{itemize}
\item[\rm a)] $C$ is called a {\it Euclidean self-dual}, resp. {\it Hermitian self-dual} code if $C=C^\perp$ resp. $C =C^{\perp_h}$.
\item[\rm b)] $C$ is {\it Euclidean  LCD}, resp. {\it Hermitian  LCD} if $C \oplus C^\perp = R$ resp. $C \oplus C^{\perp_h} = R$.
\end{itemize}
\end{defi}

Note that all twisted skew group codes are idempotent codes if $\cha K \nmid |G|$, by Theorem \ref{Maschke}. Also in the case that $R$ is not semisimple it turned out that some interesting examples are idempotent codes, for instance the ternary self-dual extended Golay code \cite{CW21}. Note that a Euclidean LCD group code is always an idempotent code \cite{CW18}. This does not hold true for a twisted LCD group code since
$C^\perp$ might not be a left ideal in $R$ for $\alpha \not= \alpha^{-1}$. \\

Suppose now that the cocycle $\alpha$
satisfies $\alpha = \alpha^{-1}$, i.e.,
$\alpha$ takes only values $\pm 1$.
In this case we have
\begin{equation} \label{eq5}
\langle a\cdot b,c \rangle= \langle a, c \cdot \widehat{b}\rangle
\end{equation} for all $a,b,c \in R$, by direct calculations.

Most of the following results  have been proved already for group or twisted group codes
in \cite{CW18,CW21}. The proofs for twisted skew group codes are obvious extenstions. Thus, for the reader's convenience, we will give only one proof.

\begin{prop}\label{LCD-idempotent} Assume that $\alpha=\alpha^{-1}$.
If $C \leq R$ is a twisted skew group code, then the following are equivalent.
\begin{itemize}
\item[\rm a)] $C$ is an Euclidean LCD code.
\item[\rm b)] $C=Re$ where $e^2=e=\widehat{e}$.
\end{itemize}
\end{prop}
\begin{proof}
First we assume b), that is $C=Re$ where $e^2=e=\widehat{e}$. Clearly, $R = Re \oplus R(1-e)$, since $e$ is an idempotent.
Using equation (\ref{eq5}) we get
$$
\langle a e , b(1-e) \rangle=\langle a ,b (1-e)\widehat{e} \rangle =\langle  a , b (1-e)e\rangle= \langle a ,0 \rangle=0,
$$ for all $a,b \in R$. Hence  $R(1-e) \leq C^{\perp}$. Since
$$
\mathrm{dim}_{K}(R(1-e))=|G|-\mathrm{dim}_{K}(C)=\mathrm{dim}_{K}(C^{\bot}),$$ we have $R(1-e)=C^{\bot}$.

Conversely assume that a) holds.  Note that $C^{\perp}$ is a left ideal of $R$, which follows from Theorem \ref{antiiso} and Theorem \ref{eucl}, using $\alpha=\alpha^{-1}$.  Let $R=C \oplus C^{\bot}$ and write $1=e+f$ with $e \in C$ and $f  \in C^{\bot}$. Clearly,
$e=e^2+ef$ and $f=fe+f^2$, which implies $e^2=e$, $f^2=f$ and $ef=fe=0$. Furthermore, $R=Re \oplus Rf$, hence $C=Re$ and $C^{\bot}=Rf$. If $a \in R$, then
$$
0=\langle ae ,f \rangle=\langle a , f \widehat{e} \rangle=\langle a , (1-e) \widehat{e} \rangle.
$$ Since $\langle \cdot \, , \cdot \rangle$ is nondegenerate, we obtain $(1-e)\widehat{e}=0$ or equivalently $\widehat{e}=e\widehat{e}$. Finally, using Theorem \ref{antiiso} again, we get
$$e=\widehat{\widehat{e}}=\widehat{e\widehat{e}}=\widehat{\widehat{e}}\widehat{e}=e\widehat{e}=\widehat{e},$$ which completes the proof.
\end{proof}

\begin{prop}
Assume that $\alpha=\alpha^{-1}$. For $0 \neq C=Re \leq R$  with $e^2=e$, the following conditions are equivalent.
\begin{itemize}
\item[\rm a)] $C$ is an Euclidean self-dual code.
\item[\rm b)] $e\widehat{e}=0$ and $\overline{1}-\widehat{e}=(\overline{1}-\widehat{e})e$. 
\end{itemize}
\end{prop}

The propositions above have the following Hermitian extension.
\begin{prop}\label{LCD-hermitian-idempotent} Assume that $\alpha=\alpha^{-1}$ and $|K|=q^2$.
If $C \leq R$ is a twisted skew group code, then the following are equivalent.
\begin{itemize}
\item[\rm a)] $C$ is an Hermitian LCD code.
\item[\rm b)] $C=Re$ where $e^2=e=\widehat{e^{(q)}}$.
\end{itemize}
\end{prop}
\begin{prop} \label{self-herm}
Assume that $\alpha=\alpha^{-1}$ and $|K|=q^2$. For $0 \neq C=Re \leq R$  with $e^2=e$, the following conditions are equivalent.
\begin{itemize}
\item[\rm a)] $C$ is an Hermitian self-dual code.
\item[\rm b)] $e\widehat{e^{(q)}}=0$ and $\overline{1}-\widehat{e^{(q)}}=(\overline{1}-\widehat{e^{(q)}})e$. 
\end{itemize}
\end{prop}

\section{Examples.}

We conclude with some examples of optimal codes in twisted skew group algebras.

\begin{ex} \label{hexa} Let  $G=D_{6} = \langle x,y \mid x^3 =y^2 = 1, x^y = x^{-1} \rangle$ be a dihedral group of order $6$.
Let $\F_4$ be the field with 4 elements and let $\omega$ be a primitive element of $\F_4$. Finally,  let $\theta$ be the non-trivial homomorphism from $D_6$ onto $\Aut(\F_4)$, i.e., $\Theta(t)$ is the Frobenius automorphism
of $\F_4$ for any involution $t\in G$ and trivial otherwise. 
We put
$$ e = \omega 1  + y +xy + \omega x^2y \in \F_4[G,\Theta,1].$$
 Then $e=e^2$ and the left ideal $C$ generated  by $ e$ in $\F_4[G,\Theta,1]$ is the $[6,3,4]$  quaternary hexacode ${\cal H}_6$. Since $C$ is generated by an idempotent, $C$ is a projective $\F_4[G,\Theta,1]$-left module.
Note that $\widehat{e^{(2)}} = e+1$.
Applying Proposition \ref{self-herm}, we see
that $C$ is Hermitian self-dual.
To our knowledge skew group codes have been considered so far only for cyclic groups.
The hexacode seems to be the first interesting example in which the group is not cyclic.

\end{ex}

\begin{rk} a) The hexacode $C={\cal H}_6$ can not be realized as a proper twisted group code: \\ Let $K=\F_4$.
Suppose that $C \leq K^\alpha G = K[G,1,\alpha]$ where  $G=D_6$ or $G=C_6$ and
$[\alpha] \not=[1]$. (Note that $D_6$ and $C_6$ are the only groups of order $6$.) If $G=D_6$, then
$[\alpha]=[1]$ since $|\HH^2(D_6,K^*)|=1$, a contradiction. Now suppose that $G=C_6$ and $C={\cal H}_6 \leq K^\alpha G$. We may assume that $\alpha = \alpha_\lambda$ as in (\ref{consta-cocycle}).
Note that $K^\alpha G$ has no irreducible module of dimension $1$, by \cite[Corollary 2.7]{w23}.
Thus $C$ is an irreducible $K^\alpha G$-module. Since $K^\alpha G$ is not semisimple by \cite[Lemma 12.2]{P}, we see that $C$ is the Jacobson radical $J(K^\alpha G)$.
If $y$ is an involution in $G$, then $(\lambda^2-y)^2=0$. Thus $\lambda^2-y$
generates a nilpotent ideal, hence
 $\lambda^2-y \in J(K^\alpha G) = C$, a contradiction, since the minimum distance of $C$ is $4$.
\\
b) The hexacode $C={\cal H}_6$ can not be realized as a  group code: \\
Let $C_6 = C_3 \times C_2 = \langle x \rangle \times \langle y \rangle$. The structure of $KG$ can be described as
$$ KG = P(V_0) \oplus P(V_1) \oplus P(V_2),$$
where $V_0= 1_G,V_1,V_2$ are the irreducible modules, all of dimension 1, and $\dim P(V_i)=2$ for all $i$. Again suppose that $C \leq KG$. If $C= V_0 \oplus V_1 \oplus V_2$, then $C$ we get a contradiction, since
$$V_i = \langle \underbrace{(1 + \lambda^ix +\lambda^{2i}x^2) +
(1 + \lambda^ix +\lambda^{2i}x^2)y}_{=v_i} \rangle$$
and $v_0 +v_1 + v_2$ has weight $2$. In the remaining case we have $P(V_i) \leq C$ for some $i$. Since all $P(V_i)$ are induced fom an ideal in $K\langle x \rangle$, the ideal $C$ contains a vector of weight $3$, a contradiction again.
Finally, we have to consider $G=D_6$. In this case we have $KG=P(1_G) \oplus V_1 \oplus V_2$, where $V_1 \cong V_2$ is irreducible of dimension $2$. As $\dim C = 3$, we see that $C$ is the direct sum of the trivial ideal and an ideal $V$ which is isomorphic to $V_1$.
Let $x$ be an element of order $3$, $y$ of order $2$, both in $G$, and $x^y = x^{-1}=x^2$. We put $N=\langle x \rangle$. Since $V$ is a projective $KN$-module and $K$ is a splitting field, we may write 
$V|_N = \langle v \rangle \oplus \langle w \rangle$ where $xv = \lambda v$ and $xw= \lambda^2 w$ and $1 \not= \lambda \in K^*$.
We write $v=v_1 + yv_2$ with $v_i \in KN$ and  may assume that the coefficient of $v_1$ is equal to $1$. 
Since $xv=\lambda v$, we get
$xv_1 = \lambda v_1$ and $xv_2 = \lambda^2 v$. Thus
$$ v= 1 + \lambda x +\lambda^2 x^2 + \mu y (1 + \lambda^2 x + \lambda x^2)$$ 
with some $\mu \in K^*$.
With $u= \mu^{-1}yv$
we get
$$ v+u = 1 + x + x^2 + (\mu +\mu^{-1})1 + (\mu\lambda^2 + \mu^{-1}\lambda)x + (\mu \lambda + \mu^{-1}\lambda^2)x^2.$$
If $\mu =1$, then $(v+u)K$ is the trivial module, a contradiction.
If $\mu = \lambda$, then $v+u$ has weight $5$, a contradiction. The same happens if $\mu = \lambda^2$, which finishes the proof.\\
c) Using group theoretic arguments, we see that $\PAut(\mathcal{H}_6)$ is a group of order 60 generated by the permutations $(1,2,6)(3,5,4)$ and $(1,2,3,4,5)$ (see for instance \cite[Example 1.7.8]{HP}). With {\sc Magma} we see that $\PAut(\mathcal{H}_6)$ does not contain a regular subgroup of order 6. Thus, according to \cite[Theorem 1.2]{BRS} and
\cite[Theorem 3,2]{CW21} the hexacode
$\mathcal{H}_6$ is neither a group code
nor a twisted group code.

\end{rk}

\begin{ex}  Let $G$ be the same group as in Example \ref{hexa} and let $\F_9$ be the field with 9 elements. We fix $\omega $  as a primitive element of $\F_9$.  Finally, let $\Theta(t)$ be again the Frobenius automorphism
of $\F_9$ for any involution $t\in G$ and trivial otherwise. With {\sc Magma}  we see that 
$$ e = \omega  +\omega^2 x+ \omega^2x^2+ x^2y \in \F_9[G,\Theta,1]$$ is a self-adjoint idempotent, i.e., $e=e^2=\hat{e}$. Therefore, by Proposition \ref{LCD-idempotent}, the left ideal $C=Re \leq R= \F_9[G,\Theta,1]$ is an Euclidean LCD skew code. Using {\sc Magma} we may check that $C$ is an MDS $[6,3,4]_9$ code with weight enumerator $W(C)=1+120x^4+240x^5+368x^6$. Note that the given construction of a $[6,3,4]_9$ code  is much simpler then
the one given in Grassl's list \cite{G}.
\end{ex}

\begin{rk} In \cite[Theorem 6.2]{W21} it is proved that Dickson's theorem holds for twisted group algebras, i.e., if $C=eK[G,1,\alpha]$ with $ 0 \not= e = e^2$ and $\cha K=p$, then $|G|_p \mid \dim C$.  As Example \ref{hexa} shows this does not hold true
in general for  skew twisted group algebras.

\end{rk}

\begin{ex}

Let $G=\Alt_4$ be the alternating group on $4$ letters, 
$K=\F_{27}=\F_3[\tau]$ with $\tau^3 + 2\tau + 1=0$, and $\Theta$ with kernel $H=C_2\times C_2\leq \textrm{Alt}(4)$. We choose 
$\alpha$ as in \cite[\S5]{CW21}. (Note that there is an obvious misprint in Table II of \cite{CW21}. We have $\alpha(za,xa^2)=-1$ and $\alpha(za,ya^2)=1$, as you can see directly from Table I.)

The following is a self-adjoint idempotent:
\begin{align*}
e=& \tau^6 {\rm Id} + \tau^{10}(1, 2)(3, 4) + \tau^{22}(1, 3, 2) + \tau^6(1, 4, 3) + \tau^{21}(2, 3, 4) + \tau^{10}(1, 2, 4) +\\
&\tau^{22}(1, 3, 4) + \tau^{23}(1, 4, 2) +
\tau^{24}(1, 2, 3) + \tau^6(1, 3)(2, 4) + \tau^{22}(1, 4)(2, 3).    
\end{align*}
Therefore by Proposition \ref{LCD-idempotent} it generates an Euclidean LCD code $C$ with parameters $[12,4,9]_{27}$ and has the following as a generator matrix
\[\gene(C)=\left[\begin{array}{cccccccccccc}
1&0&0&0&\tau^2&\tau^{17}&\tau^5&\tau^{19}&\tau^{23}&\tau^{10} & \tau^{12}&\tau^{25}\\
0&1&0&0&\tau^{11}&\tau^{12}&\tau^{15}&2&\tau^{24}&\tau^{21}&\tau^{25}&\tau^7\\
0&0&1&0&\tau^{14}&\tau^7&\tau^{14}&\tau^{10}&\tau^{12}&\tau^2&\tau^{18}&2\\
0&0&0&1&\tau&\tau&\tau^{21}&2&\tau^{12}&\tau^6&\tau^{16}&\tau^{25}\end{array}\right].\]
Note that $C$ and its dual are MDS codes. 
\end{ex}

\begin{ex} Let $G$ be the semidirect product
$G=C_7\rtimes C_3=\langle a\rangle \rtimes \langle b\rangle$ and let 
$K=\F_{8}=\F_2[\tau]$ with $\tau^3 + \tau + 1=0$. We choose $\Theta$ with kernel $H=C_7$ and 
$\alpha$ trivial. Then
\begin{align*}
c=&\tau^31 + \tau^6a + \tau a^2 + \tau^4a^3 + \tau a^4 + a^6 + \tau^4b + \tau b  a + \tau^6b  a^2 +\tau^3b  a^3 + \tau^6b  a^4 +\\
&\tau^3b  a^5 + \tau^3b  a^6 + \tau^6b^2 + \tau^2b^2  a +\tau^2b^2  a^2 + \tau^2b^2  a^4 + \tau^3b^2  a^5 + \tau^3b^2  a^6
\end{align*}
generates an LCD code $C$ with parameters $[21,14,6]_{8}$. The  dual has parameters $[21,7,12]_{8}$. Thus $C$ and its dual are optimal according to Grassl's list \cite{G}. A generator matrix of $C$ is
{\scriptsize \[\gene(C)=
\left[\begin{array}{ccccccccccccccccccccc}
1&0&0&0&0&0&0&0&0&0&0&0&0&\tau^5&0&\tau^3&0&\tau^2&\tau&\tau^4&\tau\\
0&1&0&0&0&0&0&0&0&0&0&0&0&\tau^3&0&\tau^4&\tau^6&\tau^2&1&1&\tau^3\\
0&0&1&0&0&0&0&0&0&0&0&0&0&\tau&0&\tau^3&\tau&1&\tau^6&\tau^2&\tau^3\\
0&0&0&1&0&0&0&0&0&0&0&0&0&\tau^6&0&0&\tau^4&\tau^4&\tau^4&\tau^4&\tau^4\\
0&0&0&0&1&0&0&0&0&0&0&0&0&\tau^4&0&\tau^4&1&\tau^5&0&\tau^3&\tau^2\\
0&0&0&0&0&1&0&0&0&0&0&0&0&\tau^2&0&1&\tau^4&1&\tau^5&0&1\\
0&0&0&0&0&0&1&0&0&0&0&0&0&1&0&\tau^5&\tau&\tau&0&\tau^2&0\\
0&0&0&0&0&0&0&1&0&0&0&0&0&\tau^6&0&\tau^6&\tau^6&0&\tau&\tau^4&1\\
0&0&0&0&0&0&0&0&1&0&0&0&0&\tau^5&0&\tau^2&0&1&1&0&\tau^4\\
0&0&0&0&0&0&0&0&0&1&0&0&0&\tau^4&0&0&\tau^4&\tau^2&\tau^6&\tau^6&\tau^4\\
0&0&0&0&0&0&0&0&0&0&1&0&0&\tau^3&0&\tau^4&\tau&\tau^4&\tau^3&1&\tau^5\\
0&0&0&0&0&0&0&0&0&0&0&1&0&\tau^2&0&\tau^6&\tau^4&1&0&\tau&\tau^3\\
0&0&0&0&0&0&0&0&0&0&0&0&1&\tau&0&\tau^4&\tau^4&0&\tau^4&\tau^6&\tau\\
0&0&0&0&0&0&0&0&0&0&0&0&0&0&1&\tau^4&\tau&\tau^5&\tau^2&\tau^6&\tau^3
\end{array}\right].\]}

\end{ex}

\begin{ex}
Let $G$ be the dihedral group of order $20$, $G=C_{10}\rtimes C_2=\langle a\rangle \rtimes \langle b\rangle$ and let 
$K=\F_{9}=\F_2[\tau]$ with $\tau^2 + 2\tau + 2=0$. We choose $\Theta$ with kernel $H=C_{10}$ and 
$\alpha$ as in Table \ref{tab:cocycle}. This cocycle, which is not a coboundery, has been determined thanks to the algorithm presented in \cite{EK}.

\begin{table}[]
    \centering
 {\scriptsize    \begin{tabular}{c|cccccccccccccccccccc}
&$1$&$b$&$a$&$ab$&$a^2$&$b$&$a^3$&$a^3b$&$a^4$&$a^4b$&$a^5$&$a^5b$&$a^6$&$a^6b$&$a^7$&$a^7b$&$a^8$&$a^8b$&$a^9$&$a^9b$\\
\hline
$1$&1&1&1 &1&1&1 &1&1&1 &1&1&1 &1&1&1 &1&1&1 &1&1\\
$b$&1&1&2&2&1&1&2&2&1&1&2&2&1&1&2&2&1&1&2&2\\
$a$&1&1&2&2&1&1&2&2&1&1&2&2&1&1&2&2&1&1&2&2\\
$ab$&1&1&1&1&1&1&1&1&1&1&1&1&1&1&1&1&1&1&1&1\\
$a^2$&1&1&1&1&1&1&1&1&1&1&1&1&1&1&1&1&1&1&1&1\\
$a^2b$&1&1&2&2&1&1&2&2&1&1&2&2&1&1&2&2&1&1&2&2\\
$a^3$&1&1&2&2&1&1&2&2&1&1&2&2&1&1&2&2&1&1&2&2\\
$a^3b$&1&1&1&1&1&1&1&1&1&1&1&1&1&1&1&1&1&1&1&1\\
$a^4$&1&1&1&1&1&1&1&1&1&1&1&1&1&1&1&1&1&1&1&1\\
$a^4b$&1&1&2&2&1&1&2&2&1&1&2&2&1&1&2&2&1&1&2&2\\
$a^5$&1&1&2&2&1&1&2&2&1&1&2&2&1&1&2&2&1&1&2&2\\
$a^5b$&1&1&1&1&1&1&1&1&1&1&1&1&1&1&1&1&1&1&1&1\\
$a^6$&1&1&1&1&1&1&1&1&1&1&1&1&1&1&1&1&1&1&1&1\\
$a^6b$&1&1&2&2&1&1&2&2&1&1&2&2&1&1&2&2&1&1&2&2\\
$a^7$&1&1&2&2&1&1&2&2&1&1&2&2&1&1&2&2&1&1&2&2\\
$a^7b$&1&1&1&1&1&1&1&1&1&1&1&1&1&1&1&1&1&1&1&1\\
$a^8$&1&1&1&1&1&1&1&1&1&1&1&1&1&1&1&1&1&1&1&1\\
$a^8b$&1&1&2&2&1&1&2&2&1&1&2&2&1&1&2&2&1&1&2&2\\
$a^9$&1&1&2&2&1&1&2&2&1&1&2&2&1&1&2&2&1&1&2&2\\
$a^9b$&1&1&1&1&1&1&1&1&1&1&1&1&1&1&1&1&1&1&1&1
    \end{tabular}}
    \caption{$\alpha\in \HH^2(D_{20},\F_9)$}
    \label{tab:cocycle}
 \rule{12cm}{0.5ex} 
\end{table}

Then 
\begin{align*}
 c   & =  \tau^7 + a+ \tau a^2+ \tau^2 a^3+ \tau^5 a^5+ \tau a^6+ 2 a^7+ \tau^3 a^8+ 2 a^9\\
    & +(\tau^6 + \tau a+ \tau a^2+ \tau^6 a^3+ \tau^5 a^4 + \tau^2 a^5+ \tau^7  a^6+ \tau^5 a^7+ \tau^5 a^8+ \tau^3 a^9)b
\end{align*}
generates an LCD code $C$ with parameters $[20,16,4]_9$. The dual has parameters $[20,4,15]_9$. Thus $C$ and its dual are optimal according to Grassl's list \cite{G}. A generator matrix of $C$ is  {\scriptsize  
\[\gene(C)=
\left[\begin{array}{cccccccccccccccccccc}
1&0&0&0&0&0&0&0&0&0&0&0&0&0&0&1&0&2&0&\tau\\
0&1&0&0&0&0&0&0&0&0&0&0&0&0&0&\tau^3&0&\tau&\tau^3&\tau^2\\
0&0&1&0&0&0&0&0&0&0&0&0&0&0&0&\tau^7&0&\tau^5&\tau&\tau^7\\
0&0&0&1&0&0&0&0&0&0&0&0&0&0&0&2&0&1&2&2\\
0&0&0&0&1&0&0&0&0&0&0&0&0&0&0&2&0&\tau^6&\tau^5&0\\
0&0&0&0&0&1&0&0&0&0&0&0&0&0&0&\tau^7&0&1&\tau^3&\tau^3\\
0&0&0&0&0&0&1&0&0&0&0&0&0&0&0&1&0&0&\tau&2\\
0&0&0&0&0&0&0&1&0&0&0&0&0&0&0&\tau^5&0&1&0&\tau^2\\
0&0&0&0&0&0&0&0&1&0&0&0&0&0&0&\tau^5&0&2&\tau&2\\
0&0&0&0&0&0&0&0&0&1&0&0&0&0&0&1&0&\tau^7&\tau^2&\tau\\
0&0&0&0&0&0&0&0&0&0&1&0&0&0&0&1&0&2&\tau^3&\tau^7\\
0&0&0&0&0&0&0&0&0&0&0&1&0&0&0&\tau^3&0&\tau&\tau&\tau^5\\
0&0&0&0&0&0&0&0&0&0&0&0&1&0&0&2&0&\tau^6&1&\tau^6\\
0&0&0&0&0&0&0&0&0&0&0&0&0&1&0&\tau&0&\tau^2&\tau^7&\tau^6\\
0&0&0&0&0&0&0&0&0&0&0&0&0&0&1&\tau&0&2&2&\tau\\
0&0&0&0&0&0&0&0&0&0&0&0&0&0&0&0&1&\tau^3&\tau^7&2\\
\end{array}\right].\]}
\end{ex}

\begin{rk}
We found the above optimal codes by extensive search in {\sc Magma}. It would be extremely interesting to find a method to define twisted skew $G$-codes with prescribed minimum distance, as in the classical case of BCH codes. Note that this has been already done for particular $G$-codes in \cite{B,BJ}. Moreover, note that the product between minimum distance and dimension of these examples is far away from the lower bound given in Theorem \ref{bound}.
\end{rk}

\section*{Acknowledgements}
The first author is grateful for the support of the Israel Science Foundation (grant no. 353/21). The second author was partially supported by the ANR-21-CE39-0009 - BARRACUDA (French
\emph{Agence Nationale de la Recherche}).

\end{document}